\pdfoutput=1
\documentclass{new_tlp}
\usepackage{mathptmx}

\usepackage{blindtext}
\usepackage{adjustbox}
\usepackage{subfigure}
\usepackage{wrapfig}
\usepackage{paralist}
\usepackage{multirow}
\usepackage{enumerate}
\usepackage[noline,ruled,noend]{algorithm2e}

\usepackage{amsmath,amsfonts}
\usepackage{amssymb}

\usepackage{tikz}

\usepackage{tikz-qtree}
\usetikzlibrary{positioning,shadows,arrows}

\newif\ifdraft\drafttrue
\newif\ifinlineref\inlinereffalse
\newif\iffinal\finalfalse
\newif\ifextended\extendedfalse
\newif\ifdotikz\dotikzfalse

\iffinal
\usepackage[disable]{todonotes}
\else
\fi
\ifdraft
\usepackage{todonotes}
\newcommand{\comment}[1]{{\bf\color{blue}{*** #1 ***}}}
\else
\usepackage[disable]{todonotes}
\newcommand{\comment}[1]{}
\fi

{\end{list}}

\newcounter{myenumctr}
\newenvironment{myenumerate}{\begin{list}{(\arabic{myenumctr})}{\usecounter{myenumctr}
\topsep=2pt
\setlength{\leftmargin}{7pt}
\setlength{\itemindent}{\labelwidth}
\setlength{\itemsep}{0cm}}}
{\end{list}}

\newcommand{\sys}{\ensuremath{A}}
\newcommand{\absys}{\ensuremath{\widehat{A}}}
\newcommand{\nop}[1]{}

\newcommand{\citeNselfBYB}[2]{{#1}~\citeyear{#2}}

\newcommand{\nqbls}{\vspace*{-0.25\baselineskip}}
\newcommand{\nhbls}{\vspace*{-0.5\baselineskip}}

\newcommand{\leanparagraph}[1]{\smallskip\noindent\textbf{#1}. } 

\def\mi#1{\mathit{#1\/}}

\def\ba{\begin{array}}
\def\ea{\end{array}}
\def\be{\begin{enumerate}}
\def\ee{\end{enumerate}}
\def\bi{\begin{itemize}}
\def\ei{\end{itemize}}
\def\beq{\begin{equation}}
\def\eeq#1{\label{#1}\end{equation}}
\def\beeq{\begin{equation*}}
\def\eeeq{\end{equation*}}
\def\beqq{\begin{equation*}}
\def\eeqq{\end{equation*}}

\def\ins{\,{\in}\,}

\newcommand\bcmdtab{\noindent\bgroup\tabcolsep=0pt%
  \begin{tabular}{@{}p{10pc}@{}p{20pc}@{}}}
\newcommand\ecmdtab{\end{tabular}\egroup}

  \title[Towards Abstraction in ASP with an Application on Reasoning about Agent Policies]
        {Towards Abstraction in ASP with an Application on Reasoning about Agent Policies}

  \author[Z. G. Saribatur, T. Eiter]
         {Zeynep G. Saribatur, Thomas Eiter\\
         Institute of Logic and Computation, TU Wien, Vienna, Austria\\
         \email{\{zeynep,eiter\}@kr.tuwien.ac.at}}

\jdate{March 2003}
\pubyear{2003}
\pagerange{\pageref{firstpage}--\pageref{lastpage}}
\doi{S1471068401001193}

\newtheorem{thm}{Theorem}[section]

\newtheorem{defn}{Definition}
\newtheorem{exmp}{Example}

\begin{document}

\label{firstpage}

\maketitle

  \begin{abstract}
   ASP programs are a convenient tool for problem solving, whereas with large problem instances the size of the state space can be prohibitive. We consider abstraction as a means of over-approximation and introduce a method to automatically abstract (possibly non-ground) ASP programs that preserves their structure, while reducing the size of the problem. One particular application case is the problem of defining declarative policies for reactive agents and reasoning about them, which we illustrate on examples.
  \end{abstract}


%\tableofcontents

\section{Introduction}

Answer Set Programming (ASP) is a widely used problem solving
approach. It offers declarative languages that can be used 
to formalize actions, planning, and agent policies, in an expressive setting
(e.g.\ direct and indirect action effects)
\cite{lif99c,bara-2003,DBLP:journals/aim/ErdemGL16}),
and has led to dedicated action languages \cite{lifschitz08}.
This and the availability of efficient solvers
makes ASP a convenient tool for representing and reasoning about actions.

Consider a scenario in which a robot may be in an unknown grid-cell environment with
obstacles and aim to find a missing person (Fig.~\ref{fig:scenario}). 
It acts according to a policy, which tells it
\begin{wrapfigure}{r}{3cm}

\centering

\vspace*{-.5\baselineskip}

\includegraphics[height=1.7cm]{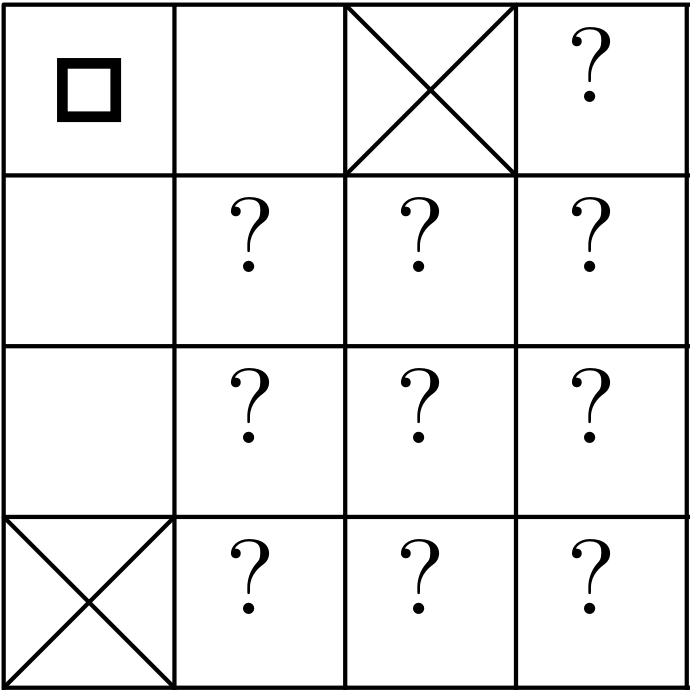}

\vspace*{-.5\baselineskip}

\caption{\quad Missing person search in an unkown environment}
\label{fig:scenario}
\vspace*{-\baselineskip}

\end{wrapfigure}
where to move next,
depending on the current observations (free / blocked cells) and
possible memory of past observations, until the person is found. To
this end, an action domain with a policy description, formalized in an
ASP program, is evaluated in each step. 
Naturally, we wonder whether the policy works, i.e.,
the person is always found, regardless of actual obstacle
locations. This can generate a large state space (for an $n{\times}n$
grid, of size larger than $2^{n{\times}n}$) and simple approaches such
as searching for a run in which the policy fails
quickly become infeasible.

To overcome this, we aim at using abstraction, which is a well-known
approach to reduce problem complexity. In a deliberate loss of
information, the problem is approximated to achieve a smaller or
simpler state space, at the price of spurious
counterexamples to the behavior \cite{clarke03}. In planning, abstraction
is mostly focused on relaxing the model, by omitting preconditions of
actions and details of the domain model
\cite{giunchiglia1992theory,knoblock1994automatically,sacerdoti1974planning}.
Cartesian abstraction \cite{seipp2013counterexample} refines
in the spirit of \cite{clarke03} failure states of abstract
trajectories, starting from a trivial abstraction; the classical
planning setting, however, disregards incomplete
initial states
%information, 
(a known source of complexity).
These works do not consider policies with background knowledge that can do decision-making with information beyond action effects.

In the area of ASP-based action languages, abstraction has not been
considered so far, and neither in the broader ASP context.
%In order to facilitate the use of a CEGAR style approach 
In order to exploit abstraction for reasoning about action
descriptions and policies in ASP, we need an abstraction method for
ASP programs that offers the following features. First, information
loss on both the model and the domain is possible. Second, 
relationships and dependencies expressed in the program should be
largely preserved. And third, abstractions should be (semi-)
automatically computable. We address this challenge with the
following contributions.

%\begin{myitemize}
\begin{itemize}
%\vspace*{-0.2\baselineskip}
\item We introduce a method to abstract ASP programs in order to
  obtain an over-approximation of the answer sets of a program
  $\Pi$. That is, a program $\Pi'$ is constructed such that each answer
  set $I$ of $\Pi$ is abstracted to some answer set $I'$ of $\Pi'$; 
  While this abstraction is many to one, {\em spurious} answer sets of
  $\Pi'$ not corresponding to any answer set of $\Pi$ may exist. 
  %may exist that do not correspond to any answer set of $\Pi\!$. %the
 %               original program $\Pi$.
%\vspace*{-0.2\baselineskip}

\item For abstraction, we consider omission of literals and also domain abstraction,
  where domain elements are merged. Note that omitting is different
  from forgetting literals (see \cite{DBLP:conf/lpnmr/Leite17} for an overview), as the latter aims at
  preserving information. %Omitting literals and domain
  %abstraction 
  The abstraction types 
  can be combined and in principle iterated to build
  hierarchical abstractions.
%\vspace*{-0.2\baselineskip}

\item The method largely preserves the structure of the rules and
  works modularly for non-ground programs. Thus, it is particularly attractive
  for abstraction of parameterized problems, as e.g.,\ in the search
  scenario (grid size $n$). Furthermore, it respects built-in
  predicates such as equality ($=$), comparisons ($<,\leq$) etc., and
  can be readily implemented, with little information on the underlying
  abstraction.
%\vspace*{-0.2\baselineskip}
\item We illustrate the use of the abstraction method for reasoning about
   actions, in particular to find counterexamples to an 
   agent policy. Here, it can be particularly useful to
   identify  and explain ``essential'' aspects of failure.
\end{itemize}
%\end{myitemize}
%\vspace*{-0.23\baselineskip}

While abstraction for ASP programs is motivated by applications in reasoning
about actions, the approach is domain independent and can be
utilized in other contexts as well.

\section{Preliminaries}

\leanparagraph{ASP} A logic program $\Pi$ is a set of rules $r$ of the form

\smallskip

\centerline{$\alpha_0 \leftarrow \alpha_1,\dots,\alpha_m,\mi{not}\
\alpha_{m+1},\dots,\mi{not}\ \alpha_n,\ \ 0\,{\leq}\, m \,{\leq}\, n,$}

\smallskip

\noindent where each $\alpha_i$ is a first-order (function-free) atom and 
$\mi{not}$
is default negation; $r$ is a \emph{constraint}
if $\alpha_0$ is falsity ($\bot$, then omitted) and a \emph{fact}  if
$n\,{=}\,0$. We also write $\alpha_0 \leftarrow B^+(r),\mi{not}\ B^-(r)$,
where $B^+(r)$ (positive body) is the set $\{\alpha_1, \dots,
\alpha_m\}$ and $B^-(r)$ (negative body) the set
$\{\alpha_{m+1},$ $\dots,\alpha_n\}$, 
or $\alpha_0 {\leftarrow} B(r)$.
Rules with variables 
%are shorthands 
stand for the set of their ground instances. 
Semantically, 
%a program 
$\Pi$ induces a set of answer sets
% (stable models)
\cite{gelfond1991classical}, which
are Herbrand models (sets $I$ of ground atoms) of $\Pi$ 
%that are 
justified
by the rules, 
%such
in that $I$ is a minimal model of $f\Pi^I=$ $\{ r \in \Pi
\mid I \models B(r)\}$
%, cf.\ 
\cite{FLP04}.
The set of answer sets of a program $\Pi$ is denoted as $AS(\Pi)$. 
Negative literals $\neg \alpha$ can be 
encoded 
%encoded as designated atoms with 
using atoms $\mi{neg}\_\alpha$ and
constraints $\leftarrow \alpha,\mi{neg}\_\alpha$.
% of stable model semantics.

Common syntactic extensions are \emph{choice rules} of the form
$\{\alpha\} \leftarrow B$, which stands for the rules $\alpha \leftarrow
B, \mi{not}\ \alpha'$ and $\alpha' \leftarrow B, \mi{not}\, \alpha$, where $\alpha'$ is a
new atom, and cardinality constraints and conditional
literals %\citeBNselfYB{Simons et al.}
\cite{simons2002extending}; in particular, $i_\ell\{\,a(X)\,{:}\,b(X)\,\}i_u$ is true whenever at least $i_\ell$ and at most $i_u$ instances of $a(X)$ subject to $b(X)$ are true.

\leanparagraph{Describing actions and states} 
ASP is used to describe dynamic domains by a ``history
program" \cite{lif99c}, 
%which is a program 
whose answer sets represent
possible evolutions of the system over a 
%(fixed)
time interval. This is
achieved by adding a time variable to the atoms, and
introducing action atoms that may cause changes
% in the system 
over time. An action is defined by its preconditions and effects over the
atoms. For illustration, the following rule describes a \emph{direct effect} of the
action $\mi{goTo(X,Y)}$ over the robot's location $rAt(X,Y)$.%, where $c(X,Y)$ represents the coordinate $X,Y$.
\nqbls
\beq \ba l
\mi{rAt}(X,Y,T{+}1) \leftarrow \mi{goTo}(X,Y,T).
%\mi{rAt}(c(X,Y),T{+}1) \leftarrow \mi{goTo}(c(X,Y),T).\\
\ea \eeq {eq:direct}
\vspace*{-1.5\baselineskip}

\noindent Actions can also have \emph{indirect effects} over the state (rules not
mentioning actions); e.g., %the robot location could allow to infer
                              %a visited cell.
the robot location  is visited:
\nqbls
\beq \ba l
\mi{visited}(X,Y,T) \leftarrow \mi{rAt}(X,Y,T).
\ea \eeq {eq:indirect}
%\vspace*{-1.5\baselineskip}
\noindent Inertia laws (unaffectedness) can be elegantly expressed, e.g.\
%\vspace*{-1.5\baselineskip}
\begin{eqnarray}
%\beq \ba {r@{~}l}
% \mi{rAt(c(X,Y),T{+}1)} &\!\!\! \leftarrow \!\!\! & \mi{rAt}(c(X,Y),T), \label{eq:inertia} \\[-3pt]
%                     &       \!\!     & \mi{not}\, \neg \mi{rAt}(c(X,Y),T{+}1).\nonumber
 \mi{rAt}(X,Y,T{+}1) \! \leftarrow \!\!\!\!\! & \mi{rAt}(X,Y,T),  \mi{not}\, \neg \mi{rAt}(X,Y,T{+}1). \nonumber
 %\label{eq:inertia}
\end{eqnarray}
%\vspace*{-1.5\baselineskip}

\noindent says that the robot location remains by default the same.

One can also give further restrictions on the state, e.g., the robot and
an obstacle can never be
% located 
in the same %
cell.
%place:
\nqbls
\beq \ba l
\bot \leftarrow \mi{rAt}(X,Y,T), \mi{obsAt}(X,Y,T).
\ea \eeq {eq:obs}
%\vspace*{-1.5\baselineskip}
%Such restrictions 
% can also be defined explicitly as

\noindent Constraints can also define \emph{preconditions} of an action, e.g.,
\nqbls
\beq \ba l
\bot \leftarrow \mi{goTo}(X,Y,T), \mi{obsAt}(X,Y,T).
\ea \eeq {eq:precond}
%\todo[inline]{also inertia laws, and obstacle's location will remain
%still}
%
Dedicated action languages carry this idea further with special 
syntax for such axioms \cite{gelfondaction98}, and can be translated to ASP \cite{giunchiglia2004nct}.

\leanparagraph{Describing a policy} In addition to defining actions as above, ASP can also be used for further reasoning about the actions by singling out some of them under certain conditions. A policy that singles out the actions to execute from the current state can be described with a set
of rules,
% that picks an action according to the state. The rules are of
where rules of
form $a {\leftarrow} B$ choose an action $a$ when certain
conditions $B$ are satisfied in the state. 
%Auxiliary rules may also
Further rules may 
% be necessary to 
describe auxiliary literals that are used by $B$. 

The rules below make the agent move towards some farthest point on the grid, unless the person is seen or caught. 
In the latter case, the agent moves towards the person's location. 
%\nhbls
\begin{equation}
 \begin{split}
&1\,\{\mi{goTo}(X1,Y1,T) : \mi{farthest}(X,Y,X1,Y1,T)\}\,1\\%[-.5ex]
& \leftarrow \mi{rAt}(X,Y,T), \mi{not\ seen}(T), \mi{not\ caught}(T).\\%[-.5ex]
%&\mi{goTo}(X1,Y1,T) {\leftarrow} \mi{rAt}(X,Y,T), \mi{farthest}(X,Y,X1,Y1,T),\\[-.5ex]
%&  \hspace{3cm} \mi{not\ seen}(T), \mi{not\ caught}(T).\\[-.5ex]
&\mi{goTo}(X,Y,T) \leftarrow \mi{seen}(T), \mi{not\ caught}(T), \mi{pAt}(X,Y,T). 
\end{split}
\label{eq:pol_formula}
\raisetag{24pt}
\end{equation} 
\nqbls

\noindent The farthest point is determined by the agent's location and the
cells considered at that state;
%thus makes it 
it is thus an indirect effect of
the previous move.
%action. The same applies for
This also applies to $\mi{seen}$ and $\mi{caught}$:
% literals.
\nqbls
\beq \begin{split}
&\mi{caught}(T) \leftarrow \mi{rAt}(X,Y,T), \mi{pAt}(X,Y,T).\\%[-.5ex]
&\mi{seen}(T) \leftarrow \mi{seeReachable}(X,Y,T), \mi{pAt}(X,Y,T).
\end{split}
\raisetag{26pt}
 \eeq {eq:pol_formula2}
Notice that above is on choosing single actions. For policies that choose a sequence of actions, the policy rules will be more involved, as the stages of the plan might have to be considered.

\section{Constructing an Abstract ASP Program}\label{sec:auto_abs}

Our aim is to over-approximate a given program through constructing a
simpler program by reducing the vocabulary and preserving the behavior
of the original program (i.e., the
results of reasoning on the original program are not lost), 
at the cost of obtaining spurious
solutions.

\begin{defn}
Given two programs $\Pi$ and $\Pi'$ with $|{\bf L}|{\geq}|{\bf L}'|$, $\Pi'$ is an \emph{abstraction} of $\Pi$ if there exists a mapping $m : {\bf L} \rightarrow {\bf L'}$ such that for $I\in AS(\Pi)$, $I'=\{m(l) ~|~ l \in I\}$ is an answer set of $\Pi'$.
\end{defn}

We consider two important base cases for an abstraction mapping $m$. Literal
omission is about omitting certain literals from the program, while
domain abstraction is on clustering different constants in the domain
and treating them as equal.

\begin{defn}
Given a program $\Pi$ and its abstraction $\Pi'$,
\begin{myenumerate}
\item $\Pi'$ is a \emph{literal omission abstraction of $\Pi$} if a
  set $L \subseteq {\bf L}$ of literals is omitted and the rest is
  kept, i.e., %${\bf L}={\bf L}'\cup L$ where 
  ${\bf L}' = {\bf L}
  \setminus L$ and $m(l)=\emptyset$ if $l\in L$ and $m(l)=l$ otherwise.
\item $\Pi'$ is a \emph{domain abstraction of $\Pi$} if there is a function $m_{d}\,{:}\,D
\,{\rightarrow}\, \widehat{D}$ for a Herbrand domain $D$ and its abstraction $\widehat{D}$, such that for $l\,{=}\,p(v_1,\dots,v_n)$ we have $m(l)\,{=}\,p(m_{d}(v_1),$ \ldots, $m_{d}(v_n))$.
\end{myenumerate}
\end{defn}

In the following sections, we show a systematic way of building an abstraction of a given ASP program.
When constructing an abstract program for a given mapping, the aim is
to ensure that every original answer set $I$ is mapped to some
abstract answer set, while (unavoidably) some spurious abstract answer
sets may be introduced. Thus, an over-approximation of the original program is achieved.
The abstraction types can be composed to obtain further abstractions.

Notice that literal omission is different
than forgetting (see \cite{DBLP:conf/lpnmr/Leite17} for an overview),
as it ensures the over-approximation of the original program by making
sure that all of the original answer sets are preserved in the
abstract program, without resorting to language extensions such as
nested logic programs that otherwise might be necessary. 

\nhbls

\subsection{Literal omission}

Given $L$, we build from $\Pi$ a program $\Pi_{\overline{L}}^m$ %according to the mapping $m$ by omitting a set $L$ of literals 
as follows. For every literal $l \in ({\bf L} \setminus L) \cup \{\bot \}$ and rule %$r: l \leftarrow B^+(r), \mi{not}\ B^-(r)$:
$r: l \leftarrow B(r)$ in $\Pi$, 
%\nqbls
\be[(1)]
\itemsep=2pt
%\item If $ B^+(r) \cup  B^-(r) \subseteq {\bf L} \setminus L$: add\\
%$m(l) \leftarrow m(B^+(r)), \mi{not}\ m(B^-(r)).$
\item if $ B(r) \subseteq {\bf L} \setminus L$, we include $m(l)
  \leftarrow m(B(r))$;
\item otherwise,  if $l\,{\neq}\,\bot$ we include for every $l' \in B(r) \cap
  L$ the rule $0 \{ m(l) \} 1 \leftarrow m(B(r)\setminus \{l'\}).$
%\item Otherwise, for every $l \in (B^+(r) \cup  B^-(r)) \cap L$, if $l\neq \bot$:\\
%add $0 \{ m(l) \} 1 \leftarrow m(B^+(r)\setminus \{l\}), \mi{not}\ m(B^-(r)\setminus \{l\}).$
\ee
Notice that constraints are omitted in the constructed program if the body contains an omitted literal. If instead, the constraint gets shrunk, then for some interpretation $\widehat{I}$, the body may fire in $\Pi_{\overline{L}}^m$, while it was not the case in $\Pi$ for any $I \in AS(\Pi)$ s.t. $m(I)=\widehat{I}$. Thus 
%$\widehat{I}$ would be not an abstract answer set, even if $I$ is an answer set for $\Pi$, and we would not have an over-approximation of $\Pi$.
$I$ cannot be mapped to an abstract answer set of $\Pi_{\overline{L}}^m$, i.e., $\Pi_{\overline{L}}^m$ is not an over-approximation of $\Pi$.

%Applying $m$ on 
Omitting non-ground literals means omitting %either
all occurrences of the predicate. %or all of them are kept. 
If in a rule $r$, the omitted non-ground literal $p(V_1,\dots,V_n)$ shares some arguments, $V_i$, with the head $l$, then $l$ is conditioned over $\mi{dom}(V_i)$ (a special predicate to represent the Herbrand domain) in the constructed rule, so that all values of $V_i$ are considered.

\begin{exmp}\label{ex:toy}
Consider the following simple program $\Pi$:
\vspace*{-.25\baselineskip}
\begin{align}
  a(X_1,X_2) &\leftarrow c(X_1), b(X_2). \label{eq:1}\\
  d(X_1,X_2) &\leftarrow a(X_1,X_2), X_1{\leq}X_2.\label{eq:2}
\end{align} 
\vspace*{-1.25\baselineskip}

In omitting $c(X)$, while rule \eqref{eq:2} remains the same, rule \eqref{eq:1} changes to 
%\beq
$0\{a(X_1,X_2)\,{:}\,\mi{dom}(X_1)\}1 \leftarrow b(X_2)$.
%, since there is no longer a condition on the value of $X_1$.
%\eeq {eq:1omit}
From %In the original program, 
$\Pi$
%, given 
and the facts $c(1),b(2)$, we get the answer set $\{c(1)$, $b(2)$,
$a(1,2)$, $d(1,2)\}$, and with $c(2),b(2)$ we get
$\{c(2)$, $b(2),$ $a(2,2)$, $d(2,2)\}$.
 After omitting $c(X)$, the abstract answer
sets with fact $b(2)$ become $\{b(2),$ $a(1,2),$ $d(1,2)\}$ and
$\{b(2)$, $a(2,2),$ $d(2,2)\}$, which cover the original answers, so that all original answer sets can be mapped to
some abstract answer set.
\end{exmp}

For a semantical more fine-grained removal, e.g., removing $c(X)$ for $X{<}3$, rules may be split in cases, e.g., (\ref{eq:1}) into $X_1{<}3$ and $X_1{\geq} 3$, and
treated after renaming separately.

The following result shows that $\Pi_{\overline{L}}^m$ can be seen as an over-approximation of $\Pi$.% \noindent $I_{|\overline{L}}$ denotes $I$ with literals in $L$ projected away, $I_{|\overline{L}} \subseteq I$.

\begin{thm}
For every $I \in AS(\Pi)$ and set $L$ of literals,
$I_{|\overline{L}} \in AS(\Pi_{\overline{L}}^m)$ where $I_{|\overline{L}} =I\setminus L$.
\end{thm}
%\nhbls

\noindent By introducing choice rules for any rule that contains the omitted literal, all possible cases that could have been achieved by having the omitted literal in the rule are covered. Thus, the abstract answer sets cover the original answer sets.

\nqbls

\subsection{Domain abstraction} 

Abstraction on the domain, $D$, divides it into equivalence classes, $\widehat{D}=\{\hat{d}_1,\dots,\hat{d}_k\}$, where some values of the variables are seen as equal. 
Such an abstraction can be constructed by keeping the structure of the
literals, and having abstract rules similar to the original ones. The
original rule may rely on certain built-in relations between the
literals' variables, e.g., $=,\neq,<,\leq$, such as \eqref{eq:2};
%it is possible to
we can automatically lift 
%these relations 
them to the abstraction
(discussed below), and 
%write e.g.\
aim to use
%\nqbls
\smallskip

\centerline{$d(\widehat{X}_1,\widehat{X}_2) \leftarrow a(\widehat{X}_1,\widehat{X}_2), \widehat{X}_1{\leq}\widehat{X}_2.$}

\smallskip

\noindent where $\widehat{X}_1,\widehat{X}_2$ are variables ranging over $\widehat{D}$.
However, due to the mapping, the lifted relations may create
uncertainties which must be dealt with. 
%For example, 
E.g.\
% if we have
for a mapping $m_d(\{1,2,3\})=k$, the atom $a(k,k)$ can be true in
the abstract state because $a(3,2)$ is true in the
original state. 
%Therefore
%Thus, directly deriving $d(k,k)$ whenever $a(k,k)$
%holds would result in losing a possible answer set. 
%there is always some $a(X_1,X_2)$ s.t. $X_1 \leq X_2$. 
%This can be avoided
 The original program can have answer sets $I$ that contain (i) $a(3,2),\mi{not}\ d(3,2)$, or (ii) $a(2,2),d(2,2)$. If we keep the structure of the original rule, in any abstract answer set $d(k,k)$ must hold if $a(k,k)$ holds; hence, no $I$ with (i) can be mapped to an abstract answer set. This would result in losing a possible answer set. 
We can avoid this
by using an altered rule
%\nhbls
\smallskip

\centerline{$0\{d(\widehat{X}_1,\widehat{X}_2)\}1 \leftarrow a(\widehat{X}_1,\widehat{X}_2), \widehat{X}_1{\leq}\widehat{X}_2.$}

\smallskip
%\nqbls

A naive approach would 
%be to 
abstract all rules by modifying the heads
to choice rules. However, %still 
negation in rule bodies may cause
a loss of
% some 
original answer sets in the abstraction. Say we have a rule with negation in the body, $d(X,X) \leftarrow \mi{not}\ a(X,X)$. If it is only changed to a choice rule in the abstract program, when $a(k,k)$ holds we will not have $d(k,k)$, while originally we can have $\{d(2,2),a(3,2)\}$. Such rules must be treated specially to catch the cases of obtaining $d(k,k)$ while $a(k,k)$ holds.

For a finer-grained and systematic approach,
we focus on rules of form $r: l \leftarrow B(r), \Gamma_{\mi{rel}}(r)$
where the variables 
in $B(r)$ are standardized apart
% in the body 
and $\Gamma_{\mi{rel}}$ consists
% of formulas with 
of built-in relation literals that
impose restrictions 
%over the variables used in the body. 
on the variables in $B(r)$.

\begin{exmp}
\label{ex:standardize}
The rules  (\ref{eq:1}) and (\ref{eq:2}) are standardized apart and they
have  $\Gamma_{\mi{rel}}(r)=\top$ (or a dummy $X\,{=}\,X$) and
$\Gamma_{\mi{rel}}(r)=X_1\leq X_2$, respectively. The rule $\mi{c} \leftarrow
\mi{r}(X,Y), \mi{p}(X,Y)$ is rewritten to the rule $\mi{c} \leftarrow
\mi{r}(X_1,Y_1), \mi{p}(X_2,Y_2),\Gamma_{rel}$ with $\Gamma_{rel}\,{=}\, (X_1\,{=}\,X_2,Y_1\,{=}\,Y_2)$.
\end{exmp}
The basic idea is as follows:
when
constructing the abstract program, we either (i) just abstract each
literal in a rule, or (ii) in case of uncertainty due to abstraction,
we guess the rule head to catch possible cases. The uncertainty may
occur 
due to having relation restrictions over non-singleton equivalence
classes
(i.e. $|m_d^{-1}(\hat{d}_i)|>1$), 
or having negative literals that are mapped to non-singleton abstract literals.

To the best of our knowledge, this is the first such approach of abstracting ASP programs.

\leanparagraph{Abstracting the relations} For simplicity, we first
focus on binary relations, e.g., $=,<,\leq,\neq$, 
and $\Gamma_{rel}(r)$ 
%which are 
of the form $\mi{rel}(X,c)$ or $\mi{rel}(X,Y)$.
%The abstract version of $\mi{rel}$ is follows:
%\beq 
%\hspace{-0.25em}(\forall\hat{d}_1,\hat{d}_2\,{\in}\,\widehat{D})
%{\mi{rel}}(\hat{d}_1,\hat{d}_2) \,{\Leftrightarrow}\, \exists x_1
%{\in} \hat{d}_1, \exists x_2 {\in} \hat{d}_2.\mi{rel}(x_1,x_2).\hspace{-1em}
%\eeq {eq:abs-rel}
%\noindent That is, ${\mi{rel}}(\hat{d}_1,\hat{d}_2)$ holds 
%if for some corresponding original values the original relation holds.%;this is easy to compute.

It is necessary to %further 
reason about the cases that can occur for the truth values of ${\mi{rel}}(\hat{d}_1,\hat{d}_2)$, for $\hat{d}_1,\hat{d}_2 \in \widehat{D}$, in order to obtain minimal abstract models that cover the original answer sets. There are four cases to consider:
%\smallskip

%\hspace{-.4cm}
%\centerline{\small
%\begin{tabular}{>{\centering}m{0.5em}>{\centering}m{4em}>{\centering}m{-1em}l}
%I & $\phantom{\neg} \mi{rel}(\hat{d}_1,\hat{d}_2)$ \ & ~$\wedge$ & $ \forall x_1 \in \hat{d}_1,\forall x_2 \in \hat{d}_2.\,\mi{rel}(x_1,x_2)$\\
%II & $\neg \mi{rel}(\hat{d}_1,\hat{d}_2)$& ~$\wedge$ & $\forall x_1 \in \hat{d}_1,\forall x_2 \in \hat{d}_2.\,\neg \mi{rel}(x_1,x_2)$\\
%III & $\phantom{\neg} \mi{rel}(\hat{d}_1,\hat{d}_2)$ &~$\wedge$ & $\exists x_1 \in \hat{d}_1,\exists x_2 \in \hat{d}_2.\, \neg \mi{rel}(x_1,x_2)$\\
%IV & $\neg \mi{rel}(\hat{d}_1,\hat{d}_2)$ &~$\wedge$ & $\exists x_1 \in
%\hat{d}_1,\exists x_2 \in \hat{d}_2.\, \mi{rel}(x_1,x_2)$
%\end{tabular}
%}
%% ZGS: above version gives errors, so for now back to enumerate version

%\be [I]
%\item $\mi{rel}(\hat{d}_1,\hat{d}_2) \wedge \forall x_1 \in \hat{d}_1,\forall x_2 \in \hat{d}_2. \mi{rel}(x_1,x_2)$
%\item $\neg \mi{rel}(\hat{d}_1,\hat{d}_2) \wedge \forall x_1 \in \hat{d}_1,\forall x_2 \in \hat{d}_2. \neg \mi{rel}(x_1,x_2)$
%\item $\mi{rel}(\hat{d}_1,\hat{d}_2) \wedge \exists x_1 \in \hat{d}_1,\exists x_2 \in \hat{d}_2. \neg \mi{rel}(x_1,x_2)$
%\item $\neg \mi{rel}(\hat{d}_1,\hat{d}_2) \wedge \exists x_1 \in \hat{d}_1,\exists x_2 \in \hat{d}_2. \mi{rel}(x_1,x_2)$
%\ee

%\medskip

\noindent\begin{tabular}{r@{}l}
 I &  $\phantom{\neg} \mi{rel}(\hat{d}_1,\hat{d}_2) \wedge \forall x_1\,{\in}\, \hat{d}_1,\forall x_2\,{\in}\, \hat{d}_2. \mi{rel}(x_1,x_2)$ \\%[1ex]
 II & $\neg \mi{rel}(\hat{d}_1,\hat{d}_2) \wedge \forall x_1\,{\in}\,  \hat{d}_1,\forall x_2\,{\in}\, \hat{d}_2. \neg \mi{rel}(x_1,x_2)$ 
     \\%[1ex]
 III  &  $\phantom{\neg} \mi{rel}(\hat{d}_1,\hat{d}_2) \wedge \exists x_1\,{\in}\, \hat{d}_1,\exists x_2\,{\in}\, \hat{d}_2. \neg \mi{rel}(x_1,x_2)$\\%[1ex]
  IV &        $\neg \mi{rel}(\hat{d}_1,\hat{d}_2) \wedge \exists x_1\,{\in}\, \hat{d}_1,\exists x_2\,{\in}\, \hat{d}_2. \mi{rel}(x_1,x_2)$
\end{tabular}

%\medskip

\noindent For $\mi{rel}(\hat{d}_1,\hat{d}_2){=} \top$, Case III is
more common in domain abstractions, while case I occurs e.g.,\ for singleton
mappings (i.e., $|\hat{d}_1|=|\hat{d}_2|=1$) or for negative relations such as $\neq$. For $\mi{rel}(\hat{d}_1,\hat{d}_2) {=}\bot$, Case II is the common case, e.g., $=,\leq$, whereas case IV may occur for negative relations or $<$.

\begin{exmp}
Consider $\mi{rel}(X,Y)=X\leq Y$ and a mapping $m_d(\{1\})=\hat{d}_1, m_d(\{2,3\})=\hat{d}_k$ with an order $\hat{d}_1 < \hat{d}_k$ on the abstract values. Notice that case I occurs for $\hat{d}_1 \leq \hat{d}_k$ and $\hat{d}_1 \leq \hat{d}_1$, while case III occurs for $\hat{d}_k \leq \hat{d}_k$. The latter is due to the possibility of having $3 \leq 2$ which is false.
\end{exmp}

The cases that the equivalence classes have for a binary $\mi{rel}$ %a given relation 
can be computed by simple queries and represented by facts of form $\mi{type}_{\mi{rel}}^{\mi{case}}(\hat{d}_1,\hat{d}_2)$ %for each binary relation $\mi{rel}$, and 
for each equivalence classes $\hat{d}_1,\hat{d}_2$.
%\todo[inline]{$\mi{type}_{\mi{rel}}(\widehat{D}_1,\dots,\widehat{D}_n,\mi{case})$ possible for complex relations}

\leanparagraph{Program abstraction} We start with a procedure for programs
with rules $r: l \leftarrow B(r), \mi{rel}(t_1,t_1')$
% that do not have more than one relation or
where $|B^-(r)|{\leq} 1$. 

For any rule $r$ and $*{\in} \{+,-\}$, let the set
$S^{*}_{\mi{rel}}(r)=\{l_j \in B^{*}(r) \mid arg(l_j) \cap
\{t_1,t_1'\} \neq \emptyset\}$ be the positive and negative literals, respectively, that share an argument with $\mi{rel}(t_1,t_1')$.
We assume for simplicity that $B^-(r) \subseteq S_{rel}(r)$
and discuss how to handle rules not meeting this assumption  later.
%it can be applied to
% more
%general rules. 

We build a program $\Pi_{dom}^m$ according to the mapping $m$ %with a domain abstraction $m_{dom}:D \rightarrow \{\widehat{D}_1,\dots,\widehat{D}_n\}$. 
as follows. 
For any rule $r: l \leftarrow B(r), \mi{rel}(t_1,t_1')$ in $\Pi$, we
add:
\vspace*{-.125\baselineskip}
{%\small
\be[\hspace{1cm}\quad(1)]

\item[(0)] If $B^+(r)\setminus S^+_{\mi{rel}}(r)\neq \emptyset$:
\vspace*{-.25\baselineskip}
\be [$~\hspace{-1.5em}$(a)]
\item If $\mi{rel}(t_1,t_1'){=} \top:$~ $m(l) \leftarrow m(B(r))$. 
\ee
\vspace*{-.25\baselineskip}

\item If $S_{\mi{rel}}^+(r)\neq \emptyset$:
\vspace*{-.35\baselineskip}
\be [$~\hspace{-1.5em}$(a)]
\itemsep=3pt
\item $m(l) \leftarrow m(B(r)), {rel}(\hat{t}_i,\hat{t}_j),\mi{type}^{\textup{I}}_{\mi{rel}}(\hat{t}_i,\hat{t}_j)$.
\item $0\{m(l)\}1 \leftarrow m(B(r)), {rel}(\hat{t}_i,\hat{t}_j),\mi{type}^{\textup{III}}_{\mi{rel}}(\hat{t}_i,\hat{t}_j)$.
\item $0\{m(l)\}1 \leftarrow m(B(r)), \neg {rel}(\hat{t}_i,\hat{t}_j),\mi{type}^{\textup{IV}}_{\mi{rel}}(\hat{t}_i,\hat{t}_j)$.
\ee
\vspace*{-.25\baselineskip}
\item If $l_i {\in} S_{\mi{rel}}^-(r)$: %apply (a-b-c) and additionally
%have $m(l) {\leftarrow} m(B(r)), {rel}(\hat{t}_i,\hat{t}_j).$, and
\vspace*{-.2\baselineskip}
\be [$~\hspace{-2em}$(a$'$)] 
\itemsep=3pt
\item $m(l) {\leftarrow} m(B(r)), {rel}(\hat{t}_i,\hat{t}_j).$
\item[(b$'$)] $0\{m(l)\}1 {\leftarrow} m(B^{\mi{shift}}_{l_i}(r)), {rel}(\hat{t}_i,\hat{t}_j), \mi{type}^{\textup{III}}_{\mi{rel}}(\hat{t}_i,\hat{t}_j).$  
\item[(c$'$)] same as (c), if $S_{\mi{rel}}^+(r){=} \emptyset$;\\[1pt]
$ 0\{m(l)\}1 {\leftarrow} m(B^{\mi{shift}}_{l_i}(r)), \neg {rel}(\hat{t}_i,\hat{t}_j),\mi{type}^{\textup{IV}}_{\mi{rel}}(\hat{t}_i,\hat{t}_j).$ \\[1pt]
$~0\{m(l)\}1 {\leftarrow} m(B^{\mi{shift}}_{l_i}(r)), {rel}(\hat{t}_i,\hat{t}_j), \mi{type}^{\textup{IV}}_{\mi{rel}}(\hat{t}_i,\hat{t}_j).$  %\\
%$0\{m(l)\}1 {\leftarrow} m(B(r)), \widehat{rel}(\widehat{D}_i,\widehat{D}_j), \mi{type}^{\textup{IV}}_{\mi{rel}}(\widehat{D}_i,\widehat{D}_j).$ 
\ee
\vspace*{-.3\baselineskip}
\ee
}
\noindent where $B^{\mi{shift}}_{l_i}(r){=}B^+(r)\cup \{l_i\},\mi{not}\ B^-(r){\setminus} \{l_i\}$.

Case (0) is the special case of having positive literals that do not share arguments with $\mi{rel}$. If $\mi{rel}{=}\top$, then it will not be processed by next steps. Thus, the abstraction of $r$ is added. The assumption on $B^-(r)$ about being included in $S_{rel}(r)$ prohibits the case $B^-(r){\setminus} S^-_{\mi{rel}}(r){\neq} \emptyset$.

%For constraints, steps that create choice rules are not applied.
If $\mi{rel}(t_1,t'_1)$ shares arguments with a positive body literal, we add
rules to grasp the possible cases resulting from the relation type. In
case of uncertainty, 
%a choice rule is added to the head, 
the head is made a choice,
and for case IV, we flip the relation, $\neg\mi{rel}$, to catch the
case of the relation holding true. If $\mi{rel}(t_1,t'_1)$ shares arguments with
a negative body literal, we need to 
%add further rules
grasp the
uncertainty arising
% due to 
from negation. We do this by adding rules in which we shift the
related literal to 
%positive, $B^{\mi{shift}}$. 
the positive body, via $B^{\mi{shift}}_{l_i}(r)$.

(2-c$'$) deals with the special case of
% having a case 
a type IV relation and a negative literal, e.g., $b(X_1) \leftarrow
\mi{not}\ a(X_1,X_2), X_1 {\neq} X_2$. If %the rule 
$r$ 
is abstracted only
by keeping the same structure, %the body of the rule 
$m(B(r))$ 
might not be
satisfied by abstract literals that actually have corresponding
literals which satisfy %the original body. 
$B(r)$. 
E.g., $a(2,3){=}\bot$
satisfies $r$; this can only be reflected in the abstraction by
% having 
$a(k,k){=}\bot$ which actually %would 
does 
not satisfy %the abstract rule body. 
$m(B(r))$. 
%Therefore 
Thus, 
when building the abstract rules, rules 
%to consider
for all combinations of shifting the literal and flipping the relation need to be added. %This way all possibilities from the original system will be covered.

Notably, the construction of $\Pi^m_{dom}$ is modular, rule by
rule;
% The rules with 
facts $p(\vec{t})$ are simply lifted to abstract facts $p(m(\vec{t}))$.

\begin{exmp}\label{ex:toy_dom}
Consider the rules from Example~\ref{ex:toy} plus 
%\nqbls
\begin{align}
%  a(X_1,X_2) &\leftarrow c(X_1), b(X_2) \label{eq:1}\\
%  d(X_1,X_2) &\leftarrow a(X_1,X_2), X_1{<}Y_2.\label{eq:2}\\
  e(X_1) &\leftarrow \mi{not}\ a(X_1,X_2), X_1{=}X_2.\footnotemark  \label{eq:3}
\end{align} 
\footnotetext{In order to ensure safety, these rules can be extended with special built-in domain predicates which do not require to be standardized apart.}
%\vspace*{-1.4\baselineskip}
%
\noindent  over the domain $D=\{1,2,3\}$.
%with input $0\{c(X){:}dom(X)\}1$, $0\{b(X){:}dom(X)\}1$.
%
Suppose $\widehat{D}{=}\{\hat{d}_1,\hat{d}_k\}$ with mapping $m_{d}(1){=}\hat{d}_1$, $m_{d}(\{2,3\}){=}\hat{d}_k$. The abstract program constructed is as follows, in simplified form:
\vspace*{-.25\baselineskip}
\begin{align}
  a(\widehat{X}_1,\widehat{X}_2) &\leftarrow c(\widehat{X}_1), b(\widehat{X}_2) \label{eq:11}\\%[-.6ex]
  d(\widehat{X}_1,\widehat{X}_2) &\leftarrow a(\widehat{X}_1,\widehat{X}_2), \widehat{X}_1 {\leq} \widehat{X}_2, \widehat{X}_1 {=} \hat{d}_1,\widehat{X}_2 {=} \hat{d}_k \!\!\label{eq:21}\\%[-.6ex]
  \hspace{-0.45cm} 0\{d(\widehat{X}_1,\widehat{X}_2)\}1 & \leftarrow a(\widehat{X}_1,\widehat{X}_2), \widehat{X}_1 {\leq} \widehat{X}_2,\widehat{X}_1{=} \hat{d}_k,\widehat{X}_2 {=} \hat{d}_k \label{eq:22}\\%[-.6ex]
  e(\widehat{X}_1) &\leftarrow \mi{not}\ a(\widehat{X}_1,\widehat{X}_2), \widehat{X}_1{=}\widehat{X}_2 \!\!\!\!\label{eq:31}\\%[-.6ex]
%  e(X_1) &\leftarrow \mi{not}\ a(X_1,X_2), X_1{=}X_2, X_1,X_2 {\in} \widehat{D}_k \!\!\!\!\label{eq:32}\\
0\{e(\widehat{X}_1)\}1 &\leftarrow a(\widehat{X}_1,\widehat{X}_2), \widehat{X}_1{=}\widehat{X}_2, \widehat{X}_1{=} \hat{d}_k,\widehat{X}_2 {=} \hat{d}_k \label{eq:33}
\end{align} 
\vspace*{-1\baselineskip}

%Above shows the constructed abstract rules with equivalence classes
%that satisfy the types. The redundant rules are omitted. 
\noindent Here the $\mi{type}^{case}_{\mi{rel}}$ literals have been evaluated, and
redundant rules are omitted.
Observe that \eqref{eq:11} is same as \eqref{eq:1} as it has
$\mi{rel}{=}\top$. 
%\comment{TE: We did not include this case in the rewriting, need to
%  deal with that: ``as $\mi{rel}(t_1,t_1)=\top$''?}
From \eqref{eq:2}, we get \eqref{eq:21} for $\hat{d}_1,\hat{d}_k$ which have case I for $\leq$, and \eqref{eq:22} for $\hat{d}_k,\hat{d}_k$ that have case III. From \eqref{eq:3}, we get \eqref{eq:31} and \eqref{eq:33} with shifting for case III.

%\comment{ZGS:shrunk the par using $AS$}
For given facts $c(3),b(2)$, %the original answer set is 
$I =\{a(3,2),e(1),e(2),e(3),c(3),b(2)\} \in AS(\Pi)$. %is one of the answer sets of the original program for the input $0\{c(X){:}dom(X)\}1$, $0\{b(X){:}dom(X)\}1$. 
After applying m, %the abstraction, 
the facts become $c(\hat{d}_k),b(\hat{d}_k)$ and %the answer set is 
$\{a(\hat{d}_k,\hat{d}_k),e(1),e(\hat{d}_k),c(\hat{d}_k),b(\hat{d}_k)\} \in AS(\Pi^m_{dom})$, which covers %the original answer set. 
$I$. Note that the choice rule \eqref{eq:33} ensures
that $e(\hat{d}_k)$ can still be obtained even 
w\textbf{}hen $a(\hat{d}_k,\hat{d}_k)$ holds. It likewise covers %the original answer set for
%the facts $c(2),b(3)$, which is $\{c(2), b(3), a(2,3), d(2,3), e(1), e(2), e(3)\}$.
$\{c(2), b(3), a(2,3), d(2,3), e(1), e(2), e(3)\} \in AS(\Pi)$ for the facts $c(2),b(3)$.

\end{exmp}

We prove that the abstraction procedure constructs a system $\Pi_{dom}^m$ that over-approximates $\Pi$.
%\comment{ZGS:modified below thm with $AS$}

\begin{thm}
Let $m$ be a domain abstraction over $\Pi$. % as described above.
% as above. 
Then for every %answer set $I$ of $\Pi$, $m(I)=\widehat{I}$ is an answer set of $\Pi_{dom}^m$.
$I \in AS(\Pi)$, $m(I) \in AS(\Pi_{dom}^m)$.
\end{thm}

\begin{proof}[Proof (sketch)]
With the 
%added 
rules (0a), (1a-1b), and (2a$'$), we ensure that
% we get
$\widehat{I}$ is a
model
of %the abstract program 
$\Pi^m_{dom}$, as we either keep the structure
of a rule $r$ or change it to a choice rule. The rules added in steps
(1b-1c) and (2b$'$-2c$'$) 
%are there
serve to catch the cases that may
% not give
violate the minimality of the model due to a negative literal or a
relation over non-singleton equivalence classes. The rules (1b,2b$'$) deal with
having a literal (resp.\ relation literal) that is false in
%the original model 
$I$ but thought to be true in the abstract model $\hat{I}$, and (1c,2c$'$)
deal with 
%having 
a literal (resp.\ relation literal) that is thought to be false in 
%the abstract model being true in the original model.
$\hat{I}$ but true in $I$.
\end{proof}
\vspace*{-0.4\baselineskip}

\leanparagraph{General case} The construction can be applied to more general programs by focusing on two aspects: 1) $|B^-(r)|{>}1$: For multiple negative literals in the rule, the shifting must be applied to each negative literal. 2) $|\Gamma_{\mi{rel}}|{>}1$: To handle multiple relation literals, a straightforward approach is to view $\Gamma_{\mi{rel}} = \mi{rel}(t_1,t'_1),$ \ldots, $\mi{rel}(t_k,t'_k)$ as 
a literal of an $n$-ary relation $\mi{rel}'(X_1,X'_1,\ldots,X_k,X'_k)$, $n\,{=}\,2k$.
The abstract version of such built-ins $\mi{rel}'$ and the type cases I-IV are readily lifted.% from
%$x_1,x_2$ in (\ref{eq:abs-rel}) to $x_1,$ \ldots, $x_n$.

%
Let $\Pi_{dom}^{*\ m}$ be the program obtained from a program $\Pi$
with the generalized abstraction procedure.
Then we 
%have the following.
obtain:
\begin{thm}
For every $I \in AS(\Pi)$% as above, the interpretation $m(I)=\widehat{I}$ is an answer set of $\Pi_{dom}^{*\ m}$.
, $m(I) \in AS(\Pi_{dom}^{*\ m})$.
\end{thm}

For constraints, the
steps creating choice rules can be skipped 
%are not needed to be applied 
as we cannot guess over $\bot$. 
Further simplifications and optimizations %can be done in order to 
can help to 
avoid introducing too many spurious answer sets.
Syntactic extensions can also be addressed. Rules with choice and cardinality constraints can be lifted with the same structure. For conditional literals with conditions over negative literals, additional rules with shifting will be necessary; otherwise, the condition can be lifted the same.

\section{Using Abstraction for Policy Refutation}\label{abs_policy}

As an application case, we are interested in the problem of defining declarative policies for reactive agents and reasoning about their behavior, especially in non-deterministic environments with uncertainty in the initial state. In such environments, searching for a plan that reaches the main goal easily becomes troublesome. Therefore, we focus on defining policies that choose a sequence of actions from the current state with the current observations, in order to achieve some subgoal, and then checking the overall behavior of these policies. More details of such policies can be found in \cite{zgs16jelia}.

\leanparagraph{Background}
%We consider 
Formally, a \emph{system} 
%is a quadruple
$\sys\,{=}\,\langle \mathcal{S},\mathcal{S}_0,\mathcal{A},\Phi\rangle$
consists of a finite set $\mathcal{S}$ of states, a set
$\mathcal{S}_0 \,{\subseteq}\, \mathcal{S}$ of initial states, a finite set
$\mathcal{A}$ of actions, and a non-deterministic transition relation
$\Phi: \mathcal{S}\,{\times}\,\mathcal{A} \rightarrow 2^{\mathcal{S}}$.

%\end{itemize}
%\noindent 
A sequence $\sigma=a_1,a_2,\dots,a_n$ of actions is \emph{executable}, if

\smallskip

\centerline{$\exists s_0,\dots,s_n \in \mathcal{S}_0\, \forall\, 0\leq i<n:
s_{i+1} \ins \Phi(s_i,a_{i+1})$}

\smallskip

\noindent holds. We denote such \emph{(potential) plans} by $\Sigma$. By
$\Sigma(s)$, resp.\ $\Phi_\Sigma(s,\sigma)$, we denote those 
%that are 
executable from $s$, 
resp.\ the set of states $s_n$ reached by them. The latter
induces the transition function
$\Phi_\Sigma: \mathcal{S} \times \Sigma \rightarrow 2^\mathcal{S}$ of
the system
$\sys_\Sigma{=}\langle \mathcal{S},\mathcal{S}_0,\Sigma,\Phi_\Sigma\rangle$.

We consider policies for a 
goal $\mu$, %(detailed below), 
that guide the agent with action sequences 
computed according to the knowledge base $KB$,
which is a formal world model in a transition system view. 

\begin{defn}[Policy]
Given a system $\sys=\langle \mathcal{S},\mathcal{S}_0,\mathcal{A},\Phi\rangle$ and a set $\Sigma$ of plans with actions of $\mathcal{A}$, a \emph{policy} is 
a function $P_{\mu, KB}:\mathcal{S} {\rightarrow} 2^{\Sigma}$
s.t. $P_{\mu, KB}(s) \,{\subseteq}\, \Sigma(s)$. 
\end{defn}
\noindent 
Assuming that $\mu$ and $KB$ are fixed, we omit subscripts of $P$.
Informally, the agent executes at state $s$ any plan $\sigma\in P_{\mu, KB}(s)$ to achieve a (hidden) subgoal, and continues at state $s'\in
\Phi_\Sigma(s,\sigma)$ by executing any plan $\sigma'\in P_{\mu,
KB}(s')$ for the next subgoal, etc, until the goal $\mu$ is established.

We focus on trajectories in
$\sys_\Sigma$
% that are
followed by the policy,
and consider goals $\mu$ expressed as propositional formulas 
over the states, which are consistent sets of propositional literals. The aim is to reach some state 
satisfying $\mu$, from all possible initial states and through all
policy trajectories;
i.e., 
%. In other words, 
the policy works if $\forall s_0\,{\in}\,\mathcal{S}_0 :
s_0\,{\models}\, {\bf AF}\mu$
holds,
where {\bf A} ranges over all trajectories under $P\!$ starting at some
initial state $s_0$ in $A_\Sigma$.

\begin{exmp}
In the missing person scenario, the goal $\mu$ can be expressed as
$\mu=\bigvee_{i,j=1}^n \mi{rAt}(i,j)\land \mi{pAt}(i,j)$, 
where $\mi{rAt}(X,Y)$ (resp.\ $\mi{pAt}(X,Y)$)  states that the 
robot (resp.\ person) is at position $X,Y$, 
or by using a designated atom $\mi{caught}$ that is defined by this formula.

\end{exmp}

\subsection{Abstraction over the policy behavior}

In order to abstract the policy behavior, we consider abstraction on the states and also on the actions; Figure~\ref{fig:abs_fig}(a) illustrates our abstraction approach. The abstract system is then built straightforwardly over the transitions 
of the policy.

\leanparagraph{State abstraction}We consider a set
%$\Psi{=}\Psi_D {\cup} \Psi_I {\cup} \Psi_r$ % = \{p_1,\dots,p_m\}$ 
$\Psi$ of abstract literals
%to describe 
for the set
$\hat{\mathcal{S}} \,{\subseteq}\, 2^{\Psi}$ of abstract states, and
an abstraction function
%
%\centerline{
$h_{st}\,{:}\,\mathcal{S} \,{\rightarrow}\, \hat{\mathcal{S}}$
%,}
%\noindent 
on states, on which further conditions might be imposed.
%, i.e., $m_D({\bf L}_{D})\neq \emptyset$.
For convenience, we denote %the abstraction
$h_{st}(s)$ %of $s\,{\in}\,\mathcal{S}$ 
%also 
by $\hat{s}$, and identify %an abstract state 
%$\hat{s} \in \widehat{\mathcal{S}}$
$\hat{s}$
 %also 
 with the states $\{ s\,{\in}\, \mathcal{S}\ {\mid}\
 h_{st}(s)\,{=}\, \hat{s}\}$
% $ \,{=}\, \{ s'\,{\in}\, \mathcal{S} \mid \hat{s'}\,{ =}\, \hat{s}\}$
abstracted to it.
% $\hat{s}$.

\begin{figure}[t]
\caption{State and action abstraction}
\vspace{0.5em}
\label{fig:abs_fig}
\centering
\subfigure[][Abstraction]{\begin{tikzpicture}
   \path (0,0) node(x) {$s$}
(0,1.5) node(y) {$\hat{s}$};
\draw[->] (x) -- node[midway,left]{$h_{st}$} (y);
\path (2,1.5) node(z) {$\hat{s}'$}
(2,0) node(t) {$s'$};
\draw[->] (z) -- node[midway,right]{$h^{-1}_{st}$} (t);
\draw[->] (y) -- node[midway,above]{$\hat{a}$} (z);
\draw[->] (x) -- node[midway,below]{$\sigma$} (t);
\node(k) at (1,0.75) {$h_{act}$};
\draw[dashed] (1,0) -- (k);
\draw[->,dashed] (k) -- (1,1.5);
\end{tikzpicture}}~~~
\subfigure[][Domain abstraction]{\includegraphics[scale=.6]{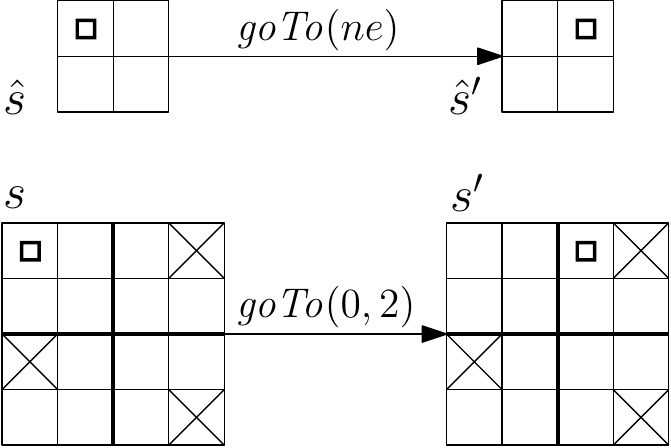}}
\vspace{-1.5em}
\end{figure}
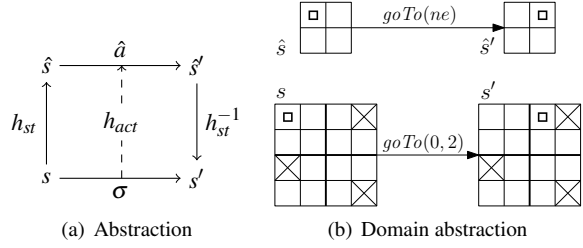

\leanparagraph{Action abstraction} 
We consider a set $\hat{\mathcal{A}}$ of abstract actions
% according to
for $\Psi$, and an abstraction function
$h_{act}{:} \Sigma {\rightarrow} \hat{\mathcal{A}}$ which maps 
action sequences $\sigma$ to abstract actions $\hat{a}$. %As
%for states, 
Similarly, 
we denote $h_{act}(\sigma)$ by $\hat{\sigma}$ and identify
$\hat{a}$ with $\{ \sigma \,{\in}\, \Sigma \mid h_{act}(\sigma)\,{=}\,\hat{a}\}$.

\begin{exmp}[ctd]\label{ex:regions}
%$\mi{rAt}(c(X,Y)), X,Y {\in} \{1,\dots,n\}$ is mapped to $rAt(c(Rx,Ry)),Rx,Ry {\in} \{1,2\}$, where 
Figure~\ref{fig:abs_fig}(b) shows an example of a domain abstraction that maps the large domain into a smaller sized domain, with $h_{st|D}(\mi{rAt}(X,Y)){=}\mi{rAt}(\lceil \frac{X}{n/2}\rceil,\lceil\frac{Y}{n/2}\rceil){=}\mi{rAt}(Rx,Ry),$ $Rx,Ry {\in} \{1,2\}$. Similarly, $\mi{goTo}(X,Y)$ is mapped to some $\mi{goTo}(Rx,Ry)$% and only affects $rAt(c(Rx,Ry))$
. For simplicity, we will refer to the abstract cells as regions $\{nw,ne,sw,se\}$.
\end{exmp}
\nhbls

\leanparagraph{Abstract system with an abstracted policy}
The transitions of
$\absys$ are defined over those 
%of the (original) system 
in $\sys$ 
%that are 
chosen by the policy.

\begin{defn}
For a system
$\sys=\langle \mathcal{S},\mathcal{S}_0,\mathcal{A},\Phi\rangle$, a set
$\Sigma$ of plans 
over $\mathcal{A}$, a transition function
$\Phi_\Sigma$, and a policy $P$, an abstract system $\widehat{A}=\langle \hat{\mathcal{S}}, \hat{\mathcal{S}}_0, \hat{\mathcal{A}}, \widehat{\Phi}_P \rangle$
is \emph{generated} by a state abstraction $h_{st}$ and 
action abstraction $h_{act}$, if 

\vspace{1pt}

\begin{compactitem}[--]
\item $\hat{\mathcal{S}}_0 = \{\hat{s}_0 \mid s_0 \in \mathcal{S}_0\}$ are the initial abstract states, and 
\item $\widehat{\Phi}_{P}: \widehat{S} \times  \hat{\mathcal{A}} \rightarrow \widehat{S}$ is the abstract transition function according to the policy $P$, defined as
%\beq 
\smallskip

\centerline{$
\widehat{\Phi}_{P}(\hat{s},\hat{a}) {=} 
\{
\hat{s}' \mid  \exists s'' \in
\hat{s}, \sigma \in P(s'')\cap \hat{a}:
s' \in \Phi_\Sigma(s'',\sigma)\}.
$}
\nqbls
\end{compactitem}
\end{defn}
Note that any abstract transition $\hat{s},\hat{a},\hat{s}'$
in $\widehat{A}$
%the original system
%$\mathcal{A}$
must 
%have a corresponding 
stem from a transition $s,\sigma,s'$ in $\mathcal{A}$
%that follows 
via policy $P$,
i.e., an abstract transition is introduced only if there is a corresponding original transition. This gives an over-approximation of the policy's behavior on the original system \cite{clarke03}.

\subsubsection{Constructing an abstract system}
We can apply the abstraction method to ASP programs 
%that have
with action
descriptions and policy rules, where we focus on 
policies with single action plans, 
with some particulars.

\begin{itemize}
 \item %Due to action abstraction $h_{act}{:} {\cal A} {\rightarrow} \hat{\cal A}$, 
 It is possible to have
the mapping $m_a{:}{\cal A}{\rightarrow} \hat{\cal A}$ %can 
create abstract
predicates, i.e., $m_a(a(v_1,\dots,v_n))=\hat{a}(m_{d}(v_1),\dots,$
$m_{d}(v_n))$ where $\hat{a} = \widehat{a(v_1,\dots,v_n)}$ depends also on the arguments of $a$ in order keep some possibly necessary details of the original actions in the abstract action.  However, with action atoms occurring in transition descriptions
only positively in rule bodies, and in policy rules only in rule
heads, no further treatment of these atoms is necessary.

\item For policy rules 
%of form 
$a {\leftarrow} B$ that select an action $a$, abstract rules
$0\{\hat{a}\}1 {\leftarrow} \widehat{B}$ 
(while correct) are undesirable as they allow to skip the action and would create a spurious trajectory. %, which needs to be refined.
To have an optimization over the abstraction, 
this can be avoided by
literals $l$ in %the body 
$B$ with singleton mappings, i.e.,
$|m^{-1}(\hat{l})|{=}1$, or with a non-singleton mapping
where $l {\in} B^{+}\setminus S^+_{\mi{rel}}$, or with cases 
that allow for simplification of the choice rules. %, e.g., every $rAt(X,Y)$ has some $\mi{farthest}(X1,Y1,X2,Y2)$ with $X{=}X1,Y{=}Y1$.
%%literals  $l {\in} B^{+}$ with $|m^{-1}(\hat{l})|{>}1$ and
%$l {\notin} S^+_{\mi{rel}}(r)$. 

%For simplicity, the time arguments are
%considered as a special term and they are lifted to be the same.
\item Time arguments amount to a special sort, and we do not abstract
over it (i.e., each time point $t$ is abstracted to itself).
Thus, time variables, terms etc.\ simply remain unaffected.

%Notice that the
\item For plan abstraction
$h_{act}{:} \Sigma {\rightarrow} \hat{\cal A}$, dedicated atoms
$\sigma$ can describe plans with their effects, obtained 
from unfolding %of the 
the 
%direct 
effect rules of the actions in $\sigma$ %give the result of executing a plan $\sigma$, and unify the preconditions of each action in $\sigma$ and define it as a precondition of $\sigma$.
%and its preconditions as the unification of the preconditions of each
%action in $\sigma$.
and their preconditions.
\end{itemize}

\begin{exmp}[ctd]
%In the initial abstraction mapping of 
%\beq \ba l
%m(rAt(X,Y,T)=rAt(A,B,T),\\ m(pAt(X,Y,T)=pAt(A,B),\\ m(caught(T))=caught(T),\\ m(seen(T))=seen(T)
%\ea \eeq {eq:abs_m}
By omitting most of the details %in the program 
except the directly
affected literals and the literals related with the goal condition,
the domain abstraction 
%described 
in Ex.\ref{ex:regions} 
%and no abstraction over the time argument gives
yields the following abstract rules
for \eqref{eq:pol_formula}-\eqref{eq:pol_formula2}: % \nhbls
\begin{align}
&1\{\mi{goTo}(\widehat{X},\widehat{Y},T){:} \widehat{X},\widehat{Y} {\in} \widehat{D}\}1 {\leftarrow} \mi{not}\ \mi{seen}(T), \mi{not}\ \mi{caught}(T).\nonumber\\%[-.5ex]
&\mi{goTo}(\widehat{X},\widehat{Y},T) \leftarrow \mi{seen}(T), \mi{not}\ \mi{caught}(T), \mi{pAt}(\widehat{X},\widehat{Y},T).\nonumber\\%[-.5ex]
&0\{\mi{caught}(T)\}1 \leftarrow \mi{rAt}(\widehat{X},\widehat{Y},T), \mi{pAt}(\widehat{X},\widehat{Y},T).\label{eq:abs_rules}\\%[-.5ex]
&0\{\mi{seen}(T)\}1 \leftarrow \mi{pAt}(\widehat{X},\widehat{Y},T). \nonumber
\end{align}
\end{exmp}
%\nhbls

\subsection{Counterexample search} Recall 
%that we started with the 
our aim of over-approxi\-mating the problem
of checking whether 
obeying the policy $P$ always reaches the goal $\mu$ 
(i.e., all paths starting from ${\cal S}_0$ %the initial states 
reach a state that satisfies $\mu$).
For policies where all
states have outgoing transitions, state abstractions that distinguish
the goal conditions can avoid false positives. 
%This means, the
%non-existence of an abstract trajectory ensures the non-existence of
%an original trajectory. 
That is, if no ``bad'' abstract trajectory exists in which
$\mu$ is unachieved (a {\em counterexample}), then no 
``bad'' original trajectory exists; % in which  $\mu$ is unachieved; 
and we can check the policy behavior on
the abstract system, as in \cite{clarke03}.

%We approach this problem by searching for a 
Concretely, we search for an abstract counterexample (cex)
trajectory in the abstract system $\absys$
% in which $\mu$ is not achieved within $k$ steps, 
of length at most $n$, where
$n$ is large enough. As the original space state $\mathcal{S}$ is
finite, any path trajectory longer than $|\mathcal{S}|+1$
%; for larger length, it would 
clearly must loop. If we cannot find such a counterexample trajectory, the policy works, cf.\ \cite{ClarkeKOS04}. 
%Our approach can be seen as planning to reach a stage where the goal is not achievable. 
%In other words, we search for an ``avoidance" path, i.e., $\exists
%\forall \neg \mu$, which may also be a loop.
%\comment{TE: Not clear why we have an $\forall$ quantifier; this is
%because of the abstracted states? But better delete it, it is not
%needed here.}
%
On the other hand,  if  a cex trajectory
$\hat{\tau}=\hat{s}_0,\hat{a}_0,\dots,\hat{s}_n$ is found, we need to
check whether $\hat{\tau}$ has a corresponding concrete trajectory $\tau$ in $\sys_\Sigma$. 
The counterexample is \emph{spurious}, if no such 
%  usch  a concrete trajectory can not be found in the original
%  system.
$\tau$ exists.

\begin{figure}[t]
\caption{}
%\vspace{-1.5em}
\centering
\subfigure[][a counterexample trajectory]{\includegraphics[scale=0.75]{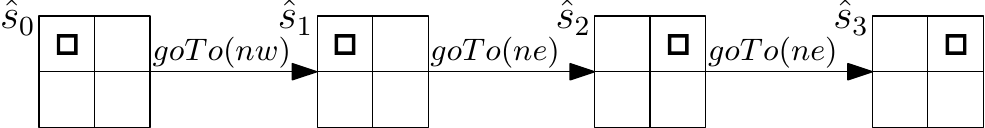}}
\subfigure[][a corresponding concrete trajectory with failure]{\includegraphics[scale=0.75]{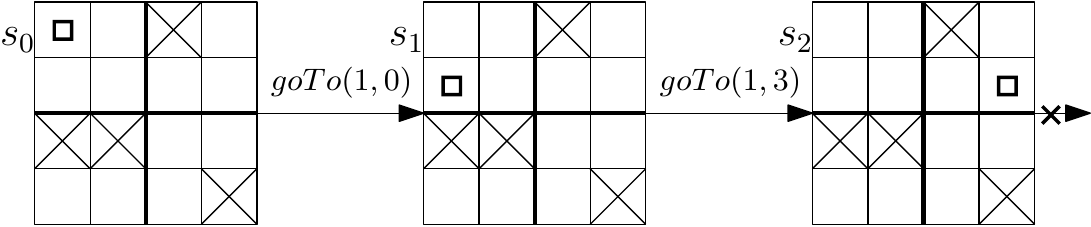}}
\label{fig:failure}

\vspace*{-1.5\baselineskip}

\end{figure}

\begin{exmp}[ctd]\label{ex:spur}
%We %describe the abstract system with the initial abstraction and then
%check for 
A 3-step counterexample (Figure~\ref{fig:failure}(a)) to
% (for simplicity) 
finding the person %, %. One
%shown in Figure~\ref{fig:failure}(a),
is 
$\hat{\tau}\,{=}\,\mi{rAt}(nw,0),\mi{goTo}(\mi{nw},0),\mi{rAt}(\mi{nw}$, $1),\mi{goTo}(\mi{ne},1),\mi{rAt(ne},2),\mi{goTo}(\mi{ne},2),\mi{rAt(ne},3)$.
Figure~\ref{fig:failure}(b) shows a corresponding trajectory %, $\tau$, 
in 
%the original system 
$\sys_\Sigma$;
% that corresponds to $\hat{\tau}$, 
it fails at step 2 to find an
%policy action 
action corresponding to
%$\mi{go(ne},\mi{ne})$, 
$\mi{goTo}(\mi{ne})$, 
as the %original 
policy 
% does not determine to stay in the same region, while there is a much
% further distanced cell.
would move to $\mi{nw}$. In fact, no corresponding trajectory 
%$\tau$ 
without
failure can be found; so $\hat{\tau}$ is spurious.
%find an action corresponding to $\mi{stayInRegion(ne})$. %So the counterexample has Type 1 failure.
\end{exmp}
 
Example~\ref{ex:spur} shows that, as expected, omitting most of the
details of the problem makes it easy to encounter spurious
trajectories. We need to add back some of the details that the policy uses in order to reduce spurious transitions.

\begin{exmp}[ctd]
Let us apply domain abstraction on $\mi{farthest}(X,Y,X1,Y1)$, and
further auxiliary literals such as $\mi{farthestDist}(X,Y,D)$. We get
$\mi{farthest}(Rx,Ry,Rx1,Ry1)$ and $\mi{farthestDist}(Rx,Ry,RD)$,
where $RD \in \{0,1\}$ tells if the distance is $<n/2$ or $\geq
n/2$. With the added back information in the abstraction, 
in the refined abstract program $\hat{\tau}$ is
no longer encountered.
%In order to
To avoid nondeterminism, %introduced by the abstract actions, 
%instead of mapping to the original predicates, we can have
%with the help of auxiliary action descriptions introduced in the
%original system.
we may use actions of the form $\mi{goWithDist}$,$\mi{goToPerson}$
introduced by  auxiliary action descriptions in the original system.
\end{exmp}

\subsection{Failure analysis and refinement}

Further studying the cause of failure in spurious cex trajectories $\hat{\tau}=\hat{s}_0,\hat{a}_0,\dots,\hat{s}_n$
would give hints on which
% omitted
information to add back in the
abstraction in order to eliminate $\hat{\tau}$. 
%In a cex trajectory, 
A failure state $\hat{s}_i$ occurs in
$\hat{\tau}$ 
due to any state $s\,{\in}\, \hat{s}_i$ that is reachable from some $s_0\,{\in}\, \hat{s}_0$ following
%the policy in 
$\sys_\Sigma$ but has no transition $\sigma$
corresponding to $\hat{a}_i$ to some state $s'\,{\in}\,\hat{s}_{i+1}$. 
%The cause of failure can be due to (i)
%the policy not determining a plan that can be abstracted to
%$\hat{a}_i$, or (ii) not having a transition to $\hat{s}_{i+1}$ with
%any of the policy plans mapped to $\hat{a}_i$.
The reason can be that (i)
$P$ determines no plan $\sigma\in P(s)$ that can be abstracted to
$\hat{a}_i$, or (ii) some such plan $\sigma$  has no transition to any
state in $\hat{s}_{i+1}$.
% with any of the policy plans mapped to $\hat{a}_i$.
%
%Some conditions can
We can impose conditions on the abstraction to avoid 
% certain 
these failures types and 
%reduce  the number of
eliminate spurious cex trajectories. However, 
%as the problem gets more complex 
if the policy uses
information depending on facts that are not affected by the
actions, a systematic approach for refinement becomes necessary.

ASP provides the possibility of reasoning over the failures and
obtaining an explanation on what is missing in the abstraction. %One
%can give constraints %that 
%the policy behavior
%must satisfy, 
Constraints about the policy behavior can be obtained,
e.g., the agent can not travel a smaller distance than
before. Depending on which constraints are violated in the failure,
one can get hints %on what information to add in the abstraction. 
for refining the abstraction. 
Thus,
a CEGAR-like approach \cite{clarke03} can be
% considered,
used, where starting from an
initial abstraction, one repeats searching for a counterexample, checks 
correctness of the latter, obtains a
failure explanation in case of spuriousness, and refines the abstraction
until a concrete counterexample is found.
\nhbls

\subsection{Discussion on evaluation}

A major difficulty of the policy checking problem is the (huge) number
        of initial states, indirect effects, choice of plans by the
        agents and possible nondeterministic effects of
        actions. Therefore existing planning approaches can not be
        used off-the-shelf. Our preliminary results on the motivating
        example show that a standard approach of searching for a
        counterexample for the policy's behavior in the original
        program is infeasible (i.e., reaches time or memory outage) 
        for larger dimensions (e.g., for $n=64$). 
        Increasing policy sophistication makes counterexample search harder, and outage is hit earlier.
		On the other hand, with an abstraction, we are able
        to easily get some candidate solution. Checking whether
        the candidate solution is concrete is faster than applying the standard approach. 
        However, in the worst case one
        needs to go through all of the introduced spurious solutions
        before reaching a correct one. Ongoing study is on dealing
        with the encountered spurious abstract solutions and to refine
        the abstraction to get rid of them.

Moreover, as a side note, even for problems such as numbering each
        node in an $n\times n$ grid with its position in a loop-free
        path (similar to the challenge planning problem VisitAll), the naive
        approach of guessing a numbering that matches the requirements
        is known not to scale, e.g., for $n=10$ no solution can be found
        in 5 hours. On the other hand, building an abstraction,
        computing and checking the abstract
        solution and (for now) manually refining the abstraction by
        hints obtained from the checking is faster to reach a concrete
        solution.

\section{Conclusion}

We introduced a    
%novel 
method for abstracting ASP programs by
over-approximating the answer sets, motivated by applications in
reasoning about actions. It keeps the structure and can be easily implemented, and further optimized.

\leanparagraph{Related Work}
In addition to %\cite{giunchiglia1992theory,knoblock1994automatically,sacerdoti1974planning,seipp2013counterexample},
abstraction in planning, 
abstraction has %also 
been studied
for agent verification 
in situation
calculus action theory \cite{banihashemi2017abstraction} %by \citeNBYB{zarriess14} and \citeNBYB{de2016verifying}, focusing
%on FO variants of LTL 
%\cite{giacomo2014ltl} and on
and of
multi-agent systems against specifications 
defined in epistemic
logic \cite{lomuscio2016verification} and alternating time temporal logic
%with recent consideration of
%abstraction 
\cite{belardinelli2016abstraction}. 
\citeNselfBYB{Banihashemi et al.}{banihashemi2017abstraction}
consider high and low-level agent specifications 
%and  defined a bisimulation between their models with the notions of
%sound and complete abstraction.%
to obtain sound and complete abstraction via bisimulation.
\citeNselfBYB{Lomuscio et al.}{lomuscio2016verification} present an automated
predicate abstraction method in three-valued
semantics, %, on which they check a given specification's  value.
%If the value can not be determined, a refinement of the abstraction is done using Craig’s interpolants \cite{belardinelli2016agent}.
and refinement based on Craig’s interpolants \cite{belardinelli2016agent}.
Abstraction has been studied in logic
programming \cite{COUSOT1992103}, but stable semantics 
%were not the focus. 
was not addressed.  

\leanparagraph{Outlook} 
Our work is a necessary starting point for
reasoning about ASP-based agent policy behavior in large environments. To
overcome
% the 
%unavoidable 
spurious answer sets, ongoing work covers
obtaining hints for a refinement 
%by making use off
from the constraint violations in
 %done by the 
spurious answers or by debugging 
the inconsistency caused by the spurious
answers approach, and a refinement methodology.

\section*{Acknowledgements}

This work has been supported by Austrian Science Fund (FWF) project W1255-N23.

%\bibliographystyle{acmtrans}
%\bibliography{ref}

\label{lastpage}
\end{document}